\documentclass[final]{siamltex}
\usepackage{amsmath, amsfonts, amssymb}

\newcommand{\rs}{\mathrm{rs}}
\newcommand{\GL}{\mathrm{GL}}

\newcommand{\Mat}{\mathrm{Mat}}
\newcommand{\F}{\mathbb{F}}
\newcommand{\G}{\mathcal{G}_q(k,n)}
\newcommand{\Vvs}{\mathcal{V}}
\newcommand{\Uvs}{\mathcal{U}}
\newcommand{\Cvs}{\mathcal{C}}

\newtheorem{example}[theorem]{Example}

\title{Pl\"ucker Embedding of Cyclic Orbit Codes\thanks{Research partially supported
    by Swiss National Science Foundation Project no. 138080.} }%
\author{Anna-Lena Trautmann\thanks{Institute of Mathematics,
  University of Zurich, Switzerland ({\tt www.math.uzh.ch/aa}).}}

\begin{document}

\maketitle

\begin{abstract}
Cyclic orbit codes are a family of constant dimension codes used for random network coding. We investigate the Pl\"ucker embedding of these codes and show how to efficiently compute the Grassmann coordinates of the code words.
\end{abstract}



\pagestyle{myheadings}
\thispagestyle{plain}
\markboth{A.-L. TRAUTMANN}{PL\"UCKER EMBEDDING OF CYCLIC ORBIT CODES}

\section{Introduction}

In network coding one is looking at the transmission of information
through a directed graph with possibly several senders and several
receivers \cite{ah00}. One can increase the throughput by linearly
combining the information vectors at intermediate nodes of the
network.  If the underlying topology of the network is unknown we
speak about \textit{random linear network coding}. Since linear spaces
are invariant under linear combinations, they are what is needed as
codewords \cite{ko08}. It is helpful (e.g. for decoding) to constrain
oneself to subspaces of a fixed dimension, in which case we talk about
\emph{constant dimension codes}.

The general linear group, consisting of all invertible transformations acts naturally on the set of all $k$-dimensional vector spaces, called the Grassmann variety. 
Orbits under this action are
called \textit{orbit codes} \cite{tr10p}. Orbit codes have useful
algebraic structure and can be seen as subspace analogues of linear block
codes in some sense \cite{tr11p}.

One can describe the balls of subspace radius $2t$ in the Grassmann variety in its Pl\"ucker embedding. Such an algebraic description of the balls of
radius $2t$ is potentially important if one is interested in an algebraic
decoding algorithm for constant dimension codes. For instance, a list decoding algorithm
requires the computation of all code words which lie in some
ball around a received message word.
In this work we characterize the Pl\"ucker embedding of orbit codes
generated by cyclic subgroups of the general linear group. 
The case of irreducible cyclic subgroups has already been studied in \cite{ro12}. This work generalizes and completes those results for general cyclic orbit codes.

The paper is structured as follows: In Section \ref{sec2} we give some preliminaries on random network coding and orbit codes in particular. Moreover, we define irreducible and completely reducible matrices and groups. Section \ref{sec3} contains the main results of this work. We first investigate the balls around subspaces of radius $2t$ with respect to the Grassmann coordinates. Then we recall how to efficiently compute these coordinates of irreducible cyclic orbit codes and show how the same can be done for reducible cyclic orbit codes. For this we distinguish between completely reducible and non-completely reducible generating groups. In Section \ref{sec4} we conclude this work.

\section{Preliminaries}\label{sec2}

Let $\mathbb{F}_q$ be the finite field with $q$ elements, where $q$ is
a prime power. We will denote the set of all $k$-dimensional subspaces of $\F_{q}^{n}$, called the Grassmannian, by $\G$ and the general linear group over $\F_{q}$ by $\GL_{n}$. Moreover, the set of all $k\times n$-matrices
over $\F_q$ is denoted by $\Mat_{k\times n}$.

Let $U\in \Mat_{k\times n}$ be a matrix of rank $k$ and
\[
\mathcal{U}=\rs (U):= \text{row space}(U)\in \G.
\]
One can notice that the row space is invariant under $\GL_k$-multiplication from
the left, i.e. for any $T\in \GL_k$ it holds that
$
\mathcal{U}=\rs(U)= \rs(T U).
$
Thus, there are several matrices that represent a given subspace. 

The \emph{subspace distance} $d_{S}$, given by
 $$ d_S(\mathcal{U},\mathcal{V}) = \dim(\Uvs)+\dim(\Vvs) - 2\dim(\mathcal{U}\cap
  \mathcal{V})
\forall \Uvs,\Vvs\in \G$$
is a metric on $\G$ and is suitable for coding over the operator channel \cite{ko08}. 
A \textit{constant dimension codes} is simply a subset of the
Grassmannian $\G$. The minimum distance is defined
in the usual way. Different constructions of constant dimension codes can be found in e.g. \cite{et08u,ko08p,ko08,ma08p,si08a,tr10p}.

Given $U\in \Mat_{k\times n}$ of rank $k$,
$\mathcal{U}\in \G$ its row space and $A\in \GL_n$, we
define
$
\mathcal{U} A:=\rs(U A).
$
This multiplication with $\GL_n$-matrices defines a
group action from the right on the Grassmannian:
\[
\begin{array}{ccc}
  \G\times \GL_n& \longrightarrow &
  \G \\ 
  (\mathcal{U},A) & \longmapsto & \mathcal{U} A
\end{array}
\]
For a
subgroup $G$ of $\GL_n$ the set
$
\mathcal{C}= \{\mathcal{U} A \mid A \in  {G}\}
$
is called an \emph{orbit code} \cite{tr10p}. 
There are different subgroups that generate the same orbit code.
An orbit code is called \emph{cyclic} if it can be defined by a cyclic
subgroup $ {G} \leq \GL_n$.


\begin{definition}
  \begin{enumerate}
  \item A matrix $G\in \GL_n$ or a subgroup $G\leq \GL_n$ is called \emph{irreducible} if 
    $\F_q^n$ contains no non-trivial $A$-invariant subspace, otherwise it is
    called \emph{reducible}.
%
  \item A matrix $G\in \GL_n$ or a subgroup $G\leq \GL_n$ is called \emph{completely reducible} if $\F_q^n$ is the direct sum of $G$-invariant subspaces which do not have any non-trivial $G$-invariant proper subspaces.
  \item  An orbit code $\mathcal{C} \subseteq \G$ is called \emph{irreducible}, respectively 
    \emph{completely reducible}, if $\mathcal{C}$ can be generated by an irreducible, respectively a completely reducible, group.
  \end{enumerate}
\end{definition}
A cyclic group is irreducible (resp. completely reducible) if and only if its generator matrix is
irreducible (resp. completely reducible). Moreover, an invertible matrix is irreducible if and only
if its characteristic polynomial is irreducible. An invertible matrix is completely reducible if and only if the blocks of its rational canonical form are companion matrices of irreducible polynomials. For more information the reader is referred to \cite{tr11p}, where one can also find that
orbit codes arising from conjugate groups are equivalent from a coding theoretic perspective. Therefore, we can restrict our studies to cyclic subgroups generated by matrices in rational canonical form.


One can describe the multiplicative action of companion matrices of irreducible polynomials via the
Galois extension field isomorphism.
 Let $p(x)=\sum p_{i}x^{i}$ be a monic irreducible polynomial over $\mathbb{F}_{q}$ of
  degree $n$ and $P$ its companion matrix
  $$P=\left(\begin{array}{ccccc} 0&1&0& \dots &0\\
  0&0&1&&0\\
  \vdots&\vdots&&\ddots&\\
  0&0&0&&1\\
  -p_{0}& -p_{1}&-p_{2}&\dots&-p_{n-1}
 \end{array}\right).$$
  Furthermore, let $\alpha
  \in \mathbb{F}_{q^n} $ be a root of $p(x)$. By $\phi^{(n)}$ we denote the
  canonical isomorphism
  \begin{align*}
    \phi^{(n)}: \quad \mathbb{F}_q^n &\longrightarrow \mathbb{F}_{q^n}\cong \F_{q}[\alpha] \\
    (v_1,\dots,v_n) &\longmapsto \sum_{i=1}^n v_i \alpha^{i-1} .
  \end{align*}
  %
The multiplication with $P$, respectively $\alpha$, commutes with $\phi^{(n)}$
\cite{ro12}, i.e. 
  for $v \in \mathbb{F}_q^n $ it holds that
  \[ \phi^{(n)}(vP) = \phi^{(n)}(v) \alpha  .
  \]

For a matrix $A\in \Mat_{m\times n}$ denote
  by $A_{i_1,\dots,i_k}$ the submatrix of $A$
  consisting of the complete rows $i_1,\dots,i_k$ and by $A[j_1,\dots,j_k]$ the submatrix of $A$ with the complete
  columns $j_1,\dots,j_k$. Consequently, $A_{i_1,\dots,i_k}[j_1,\dots,j_l]$ denotes the submatrix defined by rows $i_1,\dots,i_k$ and columns $j_1,\dots,j_l$.



  \section{Pl\"ucker Embedding of Cyclic Orbit Codes}\label{sec3}
  
Let us first recall from \cite{ro12} how to describe the balls of radius $2t$
(with respect to the subspace distance) around some $\Uvs \in \G$ with
the help of the Grassmann coordinates. 

  The map $    \varphi : \G \longrightarrow \mathbb{P}^{\binom{n}{k}-1}(\F_{q}) $ 
that maps $\rs(U)$ to
\begin{align*}
 [\det(U[1,\dots,k]) : \det(U[1,\dots,k-1,k+1]) : \dots :  \det(U[n-k+1,\dots,n])] 
\end{align*}
is an embedding of the Grassmannian in the projective space of dimension $\binom{n}{k}-1$. 
It is called the \emph{Pl\"ucker embedding} of the Grassmannian $\G$.
The projective coordinates 
\begin{align*}
&[\det(U[1,\dots,k]) : \dots : \det(U[n-k+1,\dots,n])] := \\&  
\F_{q}^{*}\cdot (\det(U[1,\dots,k]) , \dots , \det(U[n-k+1,\dots,n]))
\end{align*}
are often referred to as the \emph{Pl\"ucker} or 
\emph{Grassmann coordinates} of $\rs(U)$. For more information on the Pl\"ucker embedding of the Grassmannian the reader is referred to \cite{ho47}.

\begin{definition}
  Consider the set $\binom{[n]}{k} := \{(i_1,\dots,i_k) \mid i_l \in
  \{1,\dots,n\} \;\forall l\}$ and define a partial order on it:
 $$
  (i_{1},\dots,i_k) > (j_{1},\dots,j_{k}) \iff \exists N\in
  \mathbb{N}_{\geq 0} : i_{l}=j_{l} \;\forall l<N \textnormal{ and }
  i_{N} > j_{N} .$$
\end{definition}

Denote by $I_{k}$ the identity matrix of size $k$ and by $0_{k\times m}$ the $k \times m$-matrix with only zero entries. 
It is easy to compute the balls  of radius $2t$ around a vector space $\Uvs$, denoted by  $B_{2t}(\Uvs)$, in the following
special case.   

\begin{proposition}\cite[Proposition 3]{ro12}
Define $\Uvs_{0}:=\rs
  [\begin{array}{cc}I_{k} &0_{k\times n-k} \end{array}]$. Then
  \begin{align*}
    B_{2t}(\Uvs_{0}) = \{\Vvs \in \G \mid &
\det(\Vvs[i_1,\dots,i_k])=0 \\ 
&\;\forall (i_{1},\dots,i_{k}) \not
    \leq (t+1,\dots,k,n-t+1,\dots,n) \} .
  \end{align*}
\end{proposition}

The proposition shows that $B_{2t}(\Uvs_0)$ is described in the 
Pl\"ucker space $ \mathbb{P}^{\binom{n}{k}-1}(\F_{q})$ as a point in the
Grassmannian together with linear constraints on the Grassmann coordinates.


To derive the equations for a ball $B_{2t}(\Uvs)$ around an arbitrary 
subspace $\Uvs\in\G$ note, that for any $\Uvs\in \G$ there exists an $A\in \GL_{n}$ such that $\Uvs_{0}A=\Uvs$. Then a direct computation shows that
$$
B_{2t}(\Uvs) = B_{2t}(\Uvs_{0} A) =B_{2t}(\Uvs_{0}) A .
$$
Thus, the multiplication by $A$ transforms the linear equations
$\det(\Vvs[i_1,\dots,i_{k}]) = 0 \;\forall (i_{1},\dots,i_{k}) \not \leq
(t+1,\dots,k,n-t+1,\dots,n)$ into a new set of linear equations in the
Grassmann coordinates. 
  
Now we want to show how to explicitly compute the Grassmann coordinates of cyclic orbit codes. To do so we first recall the results for the Pl\"ucker embedding of irreducible cyclic orbit codes from \cite{ro12}. Thereafter, we will extend this theory first to completely reducible  and then to non-completely reducible orbit codes.


  \subsection{Irreducible Cyclic Orbit Codes}\label{secICOC}

 For this subsection let $\Uvs \in \G$, $p(x) = \sum_{i=0}^{n} p_{i}x^{i}
  \in \F_{q}[x]$ be a monic irreducible polynomial of degree $n$ and
  $\alpha$ a root of it. The companion matrix of $p(x)$ is denoted by
  $P$. Moreover, we define $\F_{q}^{*} := \F_{q}\setminus \{0\}$ and $\F_{q}^{*}[\alpha]:=\F_{q}[\alpha] \setminus \{0\}$.

\begin{definition}
  We define the following operation on the $k$-th outer product $ \Lambda^{k}(\F_{q}[\alpha])
  \cong \Lambda^{k}(\F_q^{n})$ :
  \begin{align*}
    * : \Lambda^{k}(\F_{q}[\alpha])\times \F_q^{*}[\alpha] 
    &\longrightarrow \Lambda^{k}(\F_{q}[\alpha])\\
    ((v_{1}\wedge \dots \wedge v_{2}), \beta) &\longmapsto (v_{1}\wedge
    \dots \wedge v_{2}) * \beta := (v_{1}\beta \wedge \dots \wedge
    v_{k}\beta) .
  \end{align*}
  This is a group action since $((v_{1}\wedge \dots \wedge v_{2}) *
  \beta) * \gamma = (v_{1}\wedge \dots \wedge v_{2}) * (\beta\gamma)$.
\end{definition}

\begin{theorem}\cite[Theorem 7]{ro12}\label{isoem}
  The following map is an embedding of the Grassmannian:
  \begin{align*}
    \varphi' : \G &\longrightarrow\mathbb{P}(\Lambda^{k}(\F_{q}[\alpha]) )\\
    \rs (U) &\longmapsto (\phi^{(n)}(U_{1}) \wedge \dots \wedge \phi^{(n)}(U_{k}))
    *\F_{q}^*
  \end{align*}
  It is isomorphic to $\varphi(\G)$ through the mapping
\begin{align*}   
& \varphi'(\rs(U))= \sum_{0\leq i_{1}<\dots<i_{k} < n} \mu_{i_{1}, \dots , i_{k}} (\alpha^{i_{1}} 
      \wedge \dots \wedge \alpha^{i_{k}})* \F_{q}^* \\  
       \longmapsto \quad & \varphi(\rs(U)) = [\mu_{0, \dots , k-1} : \dots : \mu_{n-k, \dots , n-1} ] .\end{align*}
\end{theorem}

\begin{theorem}\cite[Theorem 8]{ro12}
It holds that 
$\varphi'(\Uvs P) = \varphi'(\Uvs) * \alpha$.
Hence, an irreducible cyclic orbit code $\Cvs = \{\Uvs
P^i \mid i=0,\dots,\mathrm{ord}(P)-1\}$ has a corresponding ``Pl\"ucker
orbit'':
\[
\varphi'(\Cvs) = \{\varphi'(\Uvs) * \alpha^i \mid i=0,\dots,\mathrm{ord}(\alpha)-1\}
 = \varphi'(\Uvs) * \langle \alpha \rangle   .
\]
\end{theorem}

 \subsection{Completely Reducible Cyclic Orbit Codes}

Completely reducible cyclic subgroups of $\GL_{n}$ are exactly the ones where the blocks of the rational canonical form (RCF) of the generator matrix are companion matrices of irreducible polynomials.
Because of this property one can use the theory of irreducible cyclic orbit codes block-wise to compute the minimum distances of the block component codes and hence the minimum distance of the whole code.

For simplicity, we will explain how the theory from before generalizes in the case of generator matrices whose RCF has two blocks that are companion matrices of the irreducible polynomials. The generalization to an arbitrary number of blocks is then straightforward.

Assume our generator matrix $P$ is of the type
\[P=\left( \begin{array}{cc} P_1 & 0 \\ 0 & P_2
           \end{array}
\right)\]
where $P_1, P_2$ are companion matrices of the monic primitive polynomials $p_1(x), p_2(x)\in \F_{q}[x]$ with $\deg(p_1)=n_1, \deg(p_2)=n_2$, respectively. Let $\alpha_{1}, \alpha_{2}$ be roots of $p_{1}(x), p_{2}(x)$, respectively. Furthermore, let
\[U=\left[ \begin{array}{cc}U_1 &U_2    \end{array}
\right]\]
be a matrix representation of $\Uvs \in \G$ such that $U_1 \in Mat_{k\times n_1}, U_2\in Mat_{k\times n_2}$.
Then
\[\Uvs P^i = \rs \left[\begin{array}{cc} U_1 P_1^i &   U_2 P_2^i \end{array}\right] .\]
By $\phi^{(n_{1})}: \F_{q}^{n_{1}}\rightarrow \F_{q^{n_{1}}}\cong \F_{q}[\alpha_{1}]$ and $ \phi^{(n_{2})} : \F_{q}^{n_{2}}\rightarrow \F_{q^{n_{2}}}\cong \F_{q}[\alpha_{2}]$ we denote the standard vector space isomorphisms.

 \begin{theorem}\label{th14}\cite[Theorem 43]{tr11p}
 The map
\begin{align*}
\phi^{(n_{1},n_{2})}: \F_q^n &\longrightarrow  \F_{q}[\alpha_{1}] \times \F_{q}[\alpha_{2}]\\
(u_1,\dots,u_n) &\longmapsto (\phi^{(n_1)}(u_1,\dots,u_{n_1}), \phi^{(n_2)}(u_{n_1+1},\dots,u_n))
\end{align*}
is a vector space isomorphism. 
\end{theorem}

$\F_q^{*}[\alpha_{1}] \times \F_q^{*}[\alpha_{2}]$ forms a multiplicative group with the multiplication
\[(\beta_{1}, \beta_{2})(\gamma_{1}, \gamma_{2}) = (\beta_{1}\gamma_{1}, \beta_{2}\gamma_{2})  .\]

\begin{definition}
We define the following group action:
  \begin{align*}
    * : \Lambda^{k}(\F_{q}[\alpha_{1}]\times\F_{q}[\alpha_{2}])\times (\F_q^{*}[\alpha_{1}] \times \F_q^{*}[\alpha_{2}]) 
    &\longrightarrow \Lambda^{k}(\F_{q}[\alpha_{1}]\times\F_{q}[\alpha_{2}])\\
    ((v_{11},v_{12})\wedge \dots \wedge (v_{k1}, v_{k2}), (\beta_{1},\beta_{2})) &\longmapsto 
    ((v_{11}\beta_{1}, v_{12}\beta_{2}) \wedge \dots \wedge
    (v_{k1}\beta_{1},v_{k2}\beta_{2})) .
  \end{align*}
\end{definition}

\begin{theorem}\label{thm:main}
  The following map is an embedding of the Grassmannian:
  \begin{align*}
   \varphi^{(n_{1},n_{2})} : \G &\longrightarrow\mathbb{P}(\Lambda^{k}(\F_{q}[\alpha_{1}]\times \F_{q}[\alpha_{2}]) )\\
    \rs (U) &\longmapsto (\phi^{(n_1,n_2)}(U_{1}) \wedge \dots \wedge \phi^{{(n_1,n_2)}}(U_{k}))
    *\F_{q}^*
  \end{align*}
It is isomorphic to $\varphi(\G)$ with
\begin{align*}
((v_{11}, v_{12}) \wedge \dots \wedge (v_{k1}, v_{k2})) = \sum_{(i_1,\dots,i_k)\in\{1,2\}^k} v_{1i_1} \wedge \dots \wedge v_{ki_k}  .
\end{align*}
\end{theorem}
\begin{proof}
From Theorem \ref{isoem} we already know that $\varphi$ is an embedding. Next we show that the below defined $\psi : \varphi'(\G) \rightarrow \varphi(\G)$ is an isomorphism. Without loss of generality assume that $n_1\geq n_2$. We will use the superscript $^{(j)}$ as an additional index for the used variables and set $\lambda_{hi}^{(2)} := 0 \; \forall h,i \in \{0,\dots,n_1-1\} , i\geq n_2$. 
 \begin{align*}
   &(\phi^{(n_1,n_2)}(U_{1}) \wedge \dots \wedge \phi^{(n_1,n_2)}(U_{k})) * \F_{q}^* \\
   =& ((\sum_{i=0}^{n_1-1}\lambda_{1i}^{(1)} \alpha_1^{i}, \sum_{i=0}^{n_2-1}\lambda_{1i}^{(2)} \alpha_2^i) \wedge \dots 
   \wedge  (\sum_{i=0}^{n_1-1}\lambda_{ki}^{(1)} \alpha_1^{i}, \sum_{i=0}^{n_2-1}\lambda_{ki}^{(2)}\alpha_2^i)) * \F_{q}^* \\
    =& \sum_{(j_1,\dots,j_k)\in\{1,2\}^k}(\sum_{i=0}^{n_{j_1}-1}\lambda_{1i}^{(j_1)} \alpha_{j_1}^{i} 
    \wedge \dots \wedge \sum_{i=0}^{n_{j_k}-1}\lambda_{ki}^{(j_k)} \alpha_{j_k}^{i}) * \F_{q}^* \\
    =& \sum_{(j_1,\dots,j_k)\in\{1,2\}^k}(\sum_{0\leq i_{1},\dots,i_{k} < n_1}   \lambda_{1i_1}^{(j_1)}\dots \lambda_{ki_k}^{(j_k)}(\alpha_{j_1}^{i_1} \wedge \dots \wedge  \alpha_{j_k}^{i_k})) * \F_{q}^*  \\
    =&  \sum_{(j_1,\dots,j_k)\in\{1,2\}^k}(\sum_{0\leq i_{1}<\dots<i_{k} < n_1} \mu_{i_1,\dots, i_k}^{(j_1, \dots, j_k)}(\alpha_{j_1}^{i_1} 
    \wedge \dots \wedge  \alpha_{j_k}^{i_k})) * \F_{q}^* \\
   &\hspace{-0.3cm}\overset{\psi}{\longmapsto} [\nu_{0, \dots , k-1} : \dots : \nu_{n-k, \dots , n-1} ]
 \end{align*}
 where 
\[\mu_{i_1,\dots, i_k}^{(j_1, \dots, j_k)} := \sum_{\sigma \in S_{k}}
 (-1)^{\sigma} \lambda_{1\sigma(i_{1})}^{(j_1)} \dots  \lambda_{k\sigma(i_{k})}^{(j_k)}
 \in \F_{q}\]
and
\[\nu_{i_1+(j_1-1)n_1,\dots ,i_k+(j_k-1)n_1} := \mu_{i_1,\dots, i_k}^{(j_1, \dots, j_k)}.\]
 Since $\psi$ is an isomorphism and $\varphi' = \psi^{-1} \circ
 \varphi$, it follows that $\varphi'$ is an embedding as well.   
\end{proof}

\begin{theorem}
In analogy to the irreducible case the following holds:
\[\varphi^{(n_{1},n_{2})}(\Uvs P) = \varphi^{(n_{1},n_{2})}(\Uvs) * (\alpha_1, \alpha_2)  .\]
Hence, the code $\Cvs = \{\Uvs
P^i \mid i=0,\dots,\mathrm{ord}(P)-1\}$ has a corresponding ``Pl\"ucker
orbit''
$$
\varphi^{(n_{1},n_{2})}(\Cvs) = \{\varphi^{(n_{1},n_{2})}(\Uvs) * (\alpha_1,\alpha_2)^i \mid i=0,\dots,\mathrm{ord}(P)-1\}
 = \varphi^{(n_{1},n_{2})}(\Uvs) * \langle (\alpha_1,\alpha_2) \rangle   .
$$
\end{theorem}

%

\begin{example}
  Over $\F_2$ let $p_{1}(x)=p_{2}(x)=(x^2+x+1)$,
 $$P=\left(\begin{array}{cc|cc} 0 & 1 & 0 &0 \\ 1&1&0&0 \\\hline 0&0&0&1\\0&0&1&1  \end{array}\right) 
\quad \textnormal{  and } \quad
\Uvs = \rs\left[\begin{array}{cccc}1 & 0&0&0\\ 0&1&1&0
    \end{array}\right]  .$$
  Then $\mathrm{ord}(P)=3$ and $ \varphi^{(2,2)}(\Uvs) = ((1_{1},0_{2}) \wedge (\alpha_{1}, 1_{2}))$ where we use the subscripts to distinguish between the elements of $\F_{q}[\alpha_{1}]$ and $\F_{q}[\alpha_{2}]$. The elements of the Pl\"ucker orbit
  $\varphi^{(2,2)}(\Uvs)*\langle (\alpha_{1}, \alpha_{2}) \rangle$ are
  \begin{align*}
    & ((1_{1},0_{2}) \wedge (\alpha_{1}, 1_{2})) = (1_{1}\wedge \alpha_{1}) + (1_{1}\wedge 1_{2}) ,\\
    &((\alpha_{1},0_{2}) \wedge (\alpha_{1}+1_{1}, \alpha_{2})) = (\alpha_{1} \wedge \alpha_{1}+1_{1}) + (\alpha_{1},\alpha_{2}) = (\alpha_{1} \wedge 1_{1}) + (\alpha_{1}\wedge \alpha_{2}),\\
    &((\alpha_{1}+1_{1},0_{2}) \wedge (1_{1}, \alpha_{2}+1_{2}))= (\alpha_{1}+1_{1} \wedge 1_{1}) + (\alpha_{1}+1_{1},\alpha_{2}+1_{2})  \\
&\hspace{3.8cm}= (\alpha_{1} \wedge 1_{1}) + (\alpha_{1}\wedge \alpha_{2}) + (\alpha_{1}\wedge 1_{2}) + (1_{1}\wedge \alpha_{2}) + (1_{1}\wedge 1_{2}).
     \end{align*}
The corresponding Grassmann
  coordinates are
  \begin{align*}
     [ 1:1:0:0:0:0] ,
     [ 1:0:0:0:1:0] ,
     [ 1:1:1:1:1:0] .
  \end{align*}
One can easily verify that these are the Grassmann coordinates of the corresponding orbit code $\Uvs \langle P\rangle$.
\end{example}

 
 \subsection{Non-Completely Reducible Cyclic Orbit Codes}

But what happens if one of the blocks of the rational canonical form is not irreducible, i.e. it is the companion matrix of a reducible polynomial? For simplicity we restrict our investigation to the case that $P$ is the companion matrix of $p(x)=p_1(x)^2$ where $p_1(x)\in \F_q[x]$ is irreducible over $\F_q$ and of degree $n/2$. Together with the results of the previous subsection one can then easily generalize the results of this subsection to any other non-completely reducible generator matrix.

Consider $\F_q[x]/p(x)$ as the isomorphic representation of $\F_q^{n}$ via 
\begin{align*}
\bar{\phi}^{(n)}: \F_q^n &\longrightarrow  \F_q[x]/p(x)\\
(v_1,\dots,v_n) &\longmapsto \sum_{i=1}^{n} v_{i}x^{i-1} .
\end{align*}
In this case the multiplication with $P$ still corresponds to multiplication with $x \mod p(x)$, the difference is that $\F_q[x]/p(x)$ is not a field anymore. This is why we cannot use the $\F_q$-algebra of some $\alpha$ anymore but have to work with $\F_q[x]/p(x)$. 

Since we did not use the field structure, the definitions and results of Section \ref{secICOC} are straightforwardly carried over to this setting. Therefore the proofs of the results are omitted in this subsection.

\begin{definition}
  We define the following operation on $ \Lambda^{k}(\F_q[x]/p(x))
  \cong \Lambda^{k}(\F_q^{n})$:
  \begin{align*}
    * : \Lambda^{k}(\F_q[x]/p(x))\times (\F_q[x]/p(x))\setminus\{0\} 
    &\longrightarrow \Lambda^{k}(\F_q[x]/p(x))\\
    ((v_{1}\wedge \dots \wedge v_{k}), \beta) &\longmapsto (v_{1}\wedge
    \dots \wedge v_{k}) * \beta := (v_{1}\beta \wedge \dots \wedge
    v_{k}\beta) .
  \end{align*}
  This is a group action since $((v_{1}\wedge \dots \wedge v_{k}) *
  \beta) * \gamma = (v_{1}\wedge \dots \wedge v_{k}) * (\beta\gamma)$.
\end{definition}

\begin{theorem}
  The following map is an embedding of the Grassmannian:
  \begin{align*}
    \bar{\varphi} : \G &\longrightarrow\mathbb{P}(\Lambda^{k}(\F_q[x]/p(x)) )\\
    \rs (U) &\longmapsto (\bar{\phi}^{(n)}(U_{1}) \wedge \dots \wedge \bar{\phi}^{(n)}(U_{k}))
    *\F_{q}^*
  \end{align*}
It is isomorphic to $\varphi(\G)$.
\end{theorem}

\begin{theorem}
In analogy to the irreducible case the following holds:
\[\bar{\varphi}(\Uvs P) = \bar{\varphi}(\Uvs) * x \mod p(x)\]
Hence, the code $\Cvs = \{\Uvs
P^i \mid i=0,\dots,\mathrm{ord}(P)-1\}$ has a corresponding ``Pl\"ucker
orbit'':
\[
\bar{\varphi}(\Cvs) = \{\bar{\varphi}(\Uvs) * x^i \mid i=0,\dots,q^n-1\}
 = \bar{\varphi}(\Uvs) * \langle x \rangle
\]
\end{theorem}

\begin{example}
  Over $\F_2$ let $P$ be the companion matrix of $p(x)=(x^2+x+1)^{2}=x^4+x^2+1$ and $\Uvs \in \mathcal{G}_2(2,4)$
  such that $\bar{\phi}(\Uvs)=\{0, 1, x+x^2, 1+x+x^2\}$,
  i.e.
  \[\Uvs = \rs\left[\begin{array}{cccc}1 & 0&0&0\\ 0&1&1&0
    \end{array}\right]  .\]
  Then $ \bar{\varphi}(\Uvs) = (1 \wedge x+x^2)$. The elements of the Pl\"ucker orbit
  $\bar{\varphi}(\Uvs)*\langle x \rangle$ are
  \begin{align*}
  & (1\wedge x+x^2) = (1\wedge x) + (1\wedge x^2) \\
  & (x\wedge x^2+x^3) = (x\wedge x^2) + (x\wedge x^3)\\
  & (x^2\wedge x^3+x^4) = (x^2\wedge x^3+x^2+1) = (x^2\wedge x^3) + (x^2\wedge 1)\\
  & (x^{3}\wedge x^{4}+x^{3}+x)=(x^3\wedge x^{3}+x^{2}+x+1)= (x^3\wedge x^2) + (x^3\wedge x) +( x^3 \wedge 1)\\
  & (x^4\wedge x^{4}+x^{3}+x^{2}+x)= (x^{2}+1 \wedge x^{3}+x+1) \\
  &\hspace{3.6cm} = (x^2\wedge x^{3}) + (x^{2}\wedge x) + (x^2\wedge 1) + (1\wedge x^{3}) + (1\wedge x)\\
  & (x^{3}+x \wedge x^{4}+x^{2}+x)= (x^3+x \wedge x+1) = (x^3\wedge x) +(x^{3}\wedge 1)+ (x\wedge 1) 
     \end{align*}
The corresponding Grassmann
  coordinates are
  \begin{align*}
     [ 1:1:0:0:0:0] ,
     [ 0:0:0:1:1:0] ,
     [ 0:1:0:0:0:1] , \\
     [ 0:0:1:0:1:1] , 
     [ 1:1:1:1:0:1] , 
     [ 1:0:1:0:1:0] .
  \end{align*}
One can easily verify that these are the Grassmann coordinates of the corresponding orbit code.
\end{example}

\section{Conclusion}\label{sec4}

We described cyclic orbit codes within the Pl\"ucker space and showed that 
the orbit structure is preserved. For this we distinguished between irreducible, completely reducible and non-completely reducible generator matrices. 
Theorems \ref{isoem} and \ref{thm:main} show how this structure can be used for efficiently computing the Grassmann coordinates of the elements of a given cyclic orbit code. These coordinates can then be used for describing the balls of subspace radius $2t$ around the code words which is potentially important for new algebraic decoding algorithms.




\end{document}